\documentclass[english,aps,prl,superscriptaddress,floatfix,notitlepage,reprint]{revtex4-1}
\usepackage[T1]{fontenc}
\usepackage[latin9]{inputenc}
\usepackage{amsmath}
\usepackage{amsthm}
\usepackage{amssymb}

\makeatletter
\theoremstyle{plain}
\newtheorem{thm}{\protect\theoremname}
\theoremstyle{plain}
\newtheorem{lem}{\protect\lemmaname}
\theoremstyle{plain}
\newtheorem*{prop*}{\protect\propositionname}

\usepackage{braket}
\usepackage{txfonts} 
\usepackage{graphicx}
\usepackage{color}
\usepackage[usenames,dvipsnames]{xcolor}
\usepackage[colorlinks=true,citecolor=Blue,linkcolor=RubineRed,urlcolor=Blue]{hyperref}
\date{\today}

\makeatother

\usepackage{babel}
\newcommand{\beq}{\begin{equation}}
\newcommand{\eeq}{\end{equation}}
\newcommand{\beqa}{\begin{eqnarray}}
\newcommand{\eeqa}{\end{eqnarray}}

\def\om{\omega}

\providecommand{\lemmaname}{Lemma}
\providecommand{\propositionname}{Proposition}
\providecommand{\theoremname}{Theorem}

\begin{document}
\title{One-Dimensional Quantum Systems with Ground-State of Jastrow Form Are Integrable}

\author{Jing Yang\href{https://orcid.org/0000-0002-3588-0832} {\includegraphics[scale=0.05]{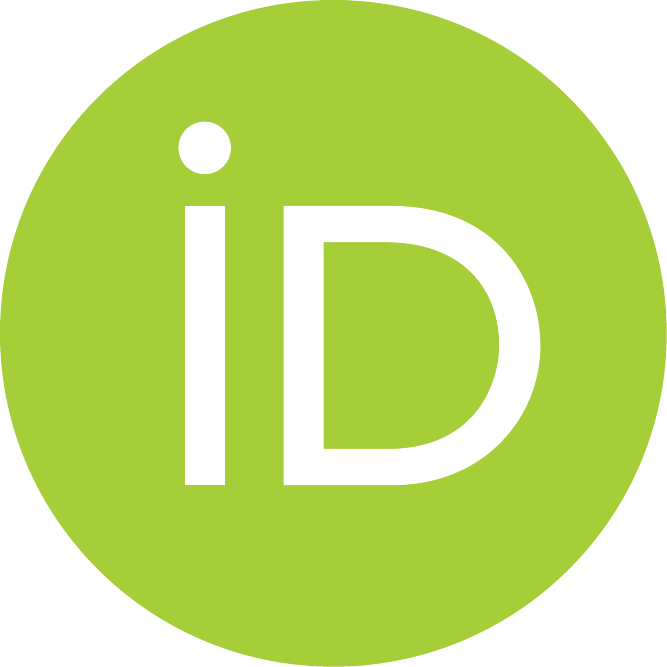}}}
\email{jing.yang@uni.lu}
\address{Department of Physics and Materials Science, University of Luxembourg,
L-1511 Luxembourg, Luxembourg}

\author{Adolfo del Campo\href{https://orcid.org/0000-0003-2219-2851}{\includegraphics[scale=0.05]{orcidid.pdf}}}
\email{adolfo.delcampo@uni.lu}
\address{Department of Physics and Materials Science, University of Luxembourg,
L-1511 Luxembourg, Luxembourg}
\address{Donostia International Physics Center, E-20018 San Sebasti\'an, Spain}

\begin{abstract}
The exchange operator formalism (EOF) 
describes many-body integrable systems using phase-space variables involving an exchange operator that acts on any pair of particles.
We establish an equivalence between models described by EOF and
the complete infinite family of parent Hamiltonians (PHJ) describing quantum many-body models with ground-states of Jastrow form. 
This makes it possible to identify the invariants of motion for any model in the PHJ family and establish its integrability, even in the presence of an external potential. 
Using this construction we establish the integrability of the long-range Lieb-Liniger model, describing bosons in a harmonic trap and subject to contact and Coulomb interactions in one dimension.  We give a variety of examples exemplifying the integrability of Hamiltonians in this family. 
\end{abstract}
\pacs{67.85.-d }
\keywords{integrable systems, ultracold gases, trapped gases in quantum fluids
and solids, strongly correlated systems}

\maketitle

Integrability in both classical and quantum many-body systems is associated with the existence of conserved quantities. 
At the quantum level, the latter correspond to operators that commute with the system Hamiltonian and govern the nonequilibrium dynamics and thermalization of a system in isolation \cite{Dziarmaga10,Mori18}.
Several integrable models have been realized in the laboratory,
prompting their use as a test-bed for quantum many-body physics, statistical
mechanics, and nonequilibrium phenomena \cite{Cazalilla11,GuanBatchelor13}. 

The integrability of a system may be proven by  finding the set of conserved
quantities. In one spatial dimension, this is possible in systems that are exactly solved using Bethe ansatz, which posits that the wavefunction of any  quantum eigenstate admits  an expansion in terms of plane waves with suitable coefficients and quasimomenta. The latter set the integrals of motion, are also known as the Bethe roots or rapidities, and serve as ``good'' quantum numbers \citep{Takahashi99,Gaudin14}.
An alternative framework is the exchange operator formalism (EOF) \citep{Polychronakos92,Polychronakos06}, in which the Hamiltonian of the quantum system admits a decoupled form in terms of generalized momenta, which readily allows for the identification of integrals of motion. {This approach can be applied to the study of excited states, as demonstrated in systems with inverse-square interactions~\cite{lapointe1996exact,ujino1996algebraic}.}
An encompassing notion of quantum integrability relies on scattering without diffraction, encoded in the
Yang-Baxter equation~\citep{yang1967someexact,baxter1978solvable,KBI97},
when collisions between particles can be described exclusively as
a sequence of two-body scattering events. 
The system is then solvable by algebraic Bethe ansatz, i.e., using the quantum inverse scattering method. 
Integrals of motion can be derived from the transfer matrix \cite{KBI97} or invoking the
asymptotic Bethe ansatz  \citep{Sutherland04,Sutherland95}. While a definite notion of quantum integrability remains under debate, many of
these approaches are closely interrelated~\citep{Caux11,Gaudin14}. In particular, EOF is related to the Yang-Baxter equation and asymptotic Bethe ansatz \citep{ujino1995correspondence,Polychronakos20}.

An important class of quantum systems is characterized by a ground-state of (Bijl-Dingle-) Jastrow form, in which
the wave function is simply the pairwise product of a pair function~\citep{Bijl40,Dingle49,Jastrow55}.
This facilitates the computation of correlation functions in these
systems \citep{Sutherland04}. The family of parent Hamiltonians with
Jastrow wave functions (PHJ, for short) can be determined by solving an inverse problem:
by acting with the kinetic energy operator in the ground-state wave
function, one can recast the resulting terms in the form of a many-body
Schr\"odinger equation, thus identifying the parent hamiltonian. 
This approach has its roots in the early works by Calogero
and Sutherland \citep{Calogero71,Sutherland71,Sutherland71pra}. It
has been extended in a number of ways \citep{CalogeroMarchioro75,Gambardella75}
and by now, for identical particles without internal degrees of freedom,
the complete family of PHJ is known both in one and
higher spatial dimensions, provided that the ground-state wave function
includes at most the product of one-particle and two-particle functions
\citep{delcampo20,BeauDC21}. The corresponding Hamiltonians generally
contain two-body and three-body interactions. It was shown by Kane
et al.~\citep{Kane91} that the three-body contribution does not affect
the low-energy physics. Further, the conditions for the three-body
term to vanish or reduce to a constant have been long-established
in the homogeneous case, in the absence of an external potential \citep{Calogero75,Sutherland04,SamajBajnok13}.

Paradigmatic instances of PHJ are integrable.
Hard-core bosons in the Tonks-Girardeau regime, realized in the laboratory
with ultracold gases~\citep{Kinoshita04,Paredes04}, have ground
state of Jastrow form \citep{Girardeau60,GWT01,GirardeauMinguzzi07}
and are integrable, being related to noninteracting fermions via the
Bose-Fermi duality \citep{Girardeau60,Girardeau04,Cazalilla11}. The
Calogero-Sutherland model with a Jastrow ground state has a harmonic
spectrum, it can be mapped to a set of independent harmonic oscillators
\citep{Kawakami93,Vacek94,Gurappa99}, and satisfies the asymptotic Bethe ansatz \citep{Sutherland95,Sutherland04}. 
Similarly, the attractive Lieb-Liniger (LL) model of bosons subject to contact interactions,  used to describe ultracold gases in tight waveguides \cite{Olshanii98,LSY03},  has a  bright quantum soliton 
as Jastrow ground-state   \citep{McGuire64}. This system  is solvable
by coordinate Bethe ansatz, which  yields the Bethe roots as integrals of motion \citep{LL63,L63,Takahashi99,Gaudin14}. 

One may thus wonder the extent to which the ground-state correlations can determine the complete integrability of the system, and what are the required conditions for this to be the case.
In this Letter, we show that the complete family of one-dimensional many-body quantum models with ground-state of Jastrow form is integrable.
To this end, we first establish the equivalence between this family and models described by EOF. In doing so, we identify explicitly the integrals of motion. Our construction holds in the presence of an external potential, which allows us to show the integrability of the long-range Lieb-Liniger model, describing bosons confined in a harmonic trap and subject to both contact and Coulomb interactions in one spatial dimension \citep{delcampo20,BeauPittman20}.

\textit{Systems described by  EOF---}
Consider the family of one-dimensional systems of identical particles without internal degrees of freedom.
It will prove useful to consider those models subject to pair-wise interactions that are possibly supplemented with three-body interactions. 
In this context, 
 EOF is a powerful framework due to Polychronakos that explicitly exhibits the integrability of a many-body quantum system in one spatial dimension \cite{Polychronakos92,Polychronakos06}. Its application has been particularly fruitful in  Calogero-Sutherland-Moser systems involving two-body inverse-square interactions  \cite{Calogero71,Sutherland71,Sutherland71b,Moser75,CalogeroMarchioro75,Sutherland04}, as discussed in \cite{Polychronakos92,Polychronakos06}. 

Let $M_{ij}$ denote the
exchange operator, which exchanges the positions
of two particles labelled by $i$ and $j$, respectively. This operator is Hermitian,
idempotent $M_{ij}^{2}=\mathbb{I}$ and symmetric with respect to
the indices, i.e., $M_{ij}=M_{ji}$. For any $1$-body operator $A_{j}\equiv A(x_{j})$,
it obeys the relations $M_{ij} A_{j}= A_{i}M_{ij}$
and $M_{ij}A_{k}=A_{k}M_{ij}$ for distinct $i$,
$j$, $k$~\citep{Polychronakos92, Polychronakos06, SM}. Note that when for spinless identical particles,
$M_{ij}$ can be identified with the permutation of two particles.
In terms of the canonical position and momentum coordinates, $x_{i}$
and $p_{j}=-\text{i}\hbar\partial/\partial x_{j}$, one can introduce the generalized
momenta 
\begin{equation}
\pi_{i}=p_{i}+\text{i}\sum_{j\neq i}V_{ij}M_{ij},\label{eq:Gmomentum}
\end{equation}
for particles $j=1,\cdots,N$. The so-called prepotential function
$V_{ij}=V(x_{i}-x_{j})$ should be antisymmetric (i.e., $V_{ij}=-V_{ji}$) 
to guarantee the Hermiticity of the generalized momenta. Using the
latter, one can construct a permutation-invariant quantities $I_{n}\equiv\sum_{i}\pi_{i}^{n}$.
In particular, $I_{2}$ is quadratic in $p_{i}$'s, and resembles
the Hamiltonian of many-body systems. To describe states of $N$ particles, consider the tensor product of the single-particle Hilbert space $\mathcal{H}$, 
i.e., $\mathcal{H}^{\otimes N}$. For indistinguishable
particles, states are restricted to the bosonic
or fermionic subspaces of $\mathcal{H}^{\otimes N}$, denoted as $\mathcal{H}_{\zeta}$
with $\zeta=+1$ for spinless bosons and $\zeta=-1$ for spinless
fermions. We define the projector $\mathcal{P}_{\zeta}$
 onto $\mathcal{H}_{\zeta}$ as~\citep{negele1998quantum}
$\mathcal{P}_{\zeta}\psi(x_{1},\,x_{2},\,\cdots,\,x_{N})=\frac{1}{N!}\sum_{\sigma}\zeta^{\sigma}\psi(x_{\sigma_{1}},\,x_{\sigma_{2}},\,\cdots,\,x_{\sigma_{N}})$,
where $\sigma$ denotes a permutation of the tuple $(1,\,2,\,\dots,\, N)$. Projecting $I_{2}$ onto the subspace $\mathcal{H}_{\zeta}$
and using 
\begin{equation}
M_{ij}\mathcal{P}_{\zeta}=\mathcal{P}_{\zeta}M_{ij}=\zeta\mathcal{P}_{\zeta},\label{eq:MP-PM-P}
\end{equation}
 we obtain $\mathcal{P}_{\zeta}I_{2}\mathcal{P}_{\zeta}/(2m)=\mathcal{P}_{\zeta}H_{0}\mathcal{P}_{\zeta}$,
where $H_{0}$ is the translation-invariant quantum many-body Hamiltonian
defined as follows
\begin{align}
H_{0} & =\sum_{i}\frac{p_{i}^{2}}{2m}+\frac{1}{m}\left[\sum_{i<j}(\zeta\hbar V'_{ij}+V_{ij}^{2})-\sum_{i<j<k}V_{ijk}\right],\label{eq:EOFH}
\end{align}
where $V_{ijk}=V_{ij}V_{jk}+V_{jk}V_{ki}+V_{ki}V_{ij}$ is fully symmetric
and a prime denotes the spatial derivative. The form of $H_{0}$ 
will play an important role in proving the integrability of the family
of Hamiltonians generated by EOF and PHJ. Specific choices of the
prepotential function $V(x)$ gives rise to well-known models. For
$V(x)=\lambda/x$, $V_{ijk}$ vanishes by permutation symmetry and
one recovers the Hamiltonian of identical particles with inverse-square
interactions \cite{Calogero71,Sutherland04}. For $V(x)=\lambda\cot(ax)$, $V_{ijk}$ is constant
and $H_{0}$ involves the inverse sine square potentials. The case
$V(x)=c\text{sgn}(x)$, corresponding $V_{ijk}$ being a negative
constant, gives rise to the celebrated Lieb-Liniger (LL) model \citep{LL63,L63}
describing ultracold gases in tight-waveguides \citep{Olshanii98,Cazalilla11}.
For all these cases where $V_{ijk}$ vanishes or is constant, $I_{n}$
commute with each other. As the system Hamiltonian coincides with
$I_{2}$ on the bosonic or fermionic sector, the set of $I_{n}$ can
be identified as invariants of motion, i.e., $[I_{n},I_{m}]=0$. We
note that all the models that have been shown to be integrable by
the EOF in Ref.~\citep{Polychronakos92} happen to have a ground-state
wave function of Jastrow form, which we discuss next.

\textit{Parent Hamiltonians with Jastrow ground-state.---} Consider
a homogenous one-dimensional many-body quantum system described by
a ground state of Jastrow form~\citep{Bijl40,Dingle49,Jastrow55},
\begin{equation}
\Phi_{0}(x_{1},\dots,x_{N})=\prod_{i<j}f_{ij},\label{eq:Phi0free}
\end{equation}
this is, the pairwise product of the pair function $f_{ij}=f(x_{i}-x_{j})$~\citep{Sutherland04}.
In the ``beautiful models''~\citep{Sutherland04} that concern us here, quantum statistics
is encoded in the symmetry of $f(x)$ which is an even function for
bosons and odd for fermions, i.e., without resorting to the use of
permanents or determinants. The case of one dimensional anyons can
similarly be taken into account by including a phase factor $\theta$, i.e.,
$f(x)= e^{-i\theta}f(-x)$ \citep{Kundu99,Girardeau06,Batchelor06}.
The complete family of PHJ of ground-state of the Jastrow form~(\ref{eq:Phi0free})
has been identified in one spatial dimension \citep{delcampo20}
and includes paradigmatic models such as the LL gas with contact interactions
\citep{LL63,L63} and the rational Calogero-Sutherland model with
inverse-square interactions \citep{Calogero71,Sutherland71}, as well
as the recently introduced long-range LL model \citep{BeauPittman20}.
For a given choice of $f$, the parent Hamiltonian $H_{0}$ takes
the form
\begin{equation}
H_{0}=\sum_{i}\frac{p_{i}^{2}}{2m}+\frac{\hbar^{2}}{m}\left[\sum_{i<j}\frac{f_{ij}''}{f_{ij}}+\sum_{i<j<k}\left(\frac{f_{ij}'f_{ik}'}{f_{ij}f_{ik}}-\frac{f_{ij}'f_{jk}'}{f_{ij}f_{jk}}+\frac{f_{ik}'f_{jk}'}{f_{ik}f_{jk}}\right)\right].\label{eq:PHJH}
\end{equation}
Here, $f'$ and $f''$ denote the first and second spatial derivatives
of $f$, respectively. The explicit expressions for this Hamiltonian
directly follow from evaluating the Laplacian on the Jastrow wave
function~(\ref{eq:Phi0free}) and recasting all resulting terms in
the form of a Schr\"odinger equation.

\textit{Equivalence of EOF and PHJ for spinless indistinguishable
particles.---} We now establish the correspondence between EOF and
PHJ for spinless identical particles. Comparing the
EOF Hamiltonian and the PHJ  in Eqs.~(\ref{eq:EOFH}) and (\ref{eq:PHJH}),
 the two-body terms are equal if
$\hbar^{2}f''(x_{ij})/f(x_{ij})=\zeta\hbar V'_{ij}+V_{ij}^{2}.$ Thus,
the prepotential reads
\begin{equation}
V_{ij}=\zeta\hbar\frac{d}{dx_{ij}}\log(f_{ij})=\zeta\hbar\frac{f_{ij}'}{f_{ij}}.
\label{eq:preV}
\end{equation}
Independently of whether the pair function is symmetric or antisymmetric,
its logarithmic derivative is guaranteed to be odd $f_{ij}'/f_{ij}=-f_{ji}'/f_{ji}$.
Thus, this property holds for spinless bosons and fermions. 
The antisymmetry of the prepotential in Eq.~(\ref{eq:preV}) guarantees the Hermiticity
condition of the associated generalized momenta in the EOF, 
\begin{equation}
\pi_{i}=p_{i}+\text{i}\zeta\hbar\sum_{j\neq i}\frac{f_{ij}'}{f_{ij}}M_{ij}\label{eq:Gpi-f}.
\end{equation}
The prepotential in Eq.~(\ref{eq:preV}) further ensures the equivalence of the three-body interaction in the
EOF and the PHJ. Thus,
any spinless system described by EOF, as in Eq.~(\ref{eq:EOFH}), has a ground-state
of Jastrow form with a pair function $f_{ij}=\exp[\int^{x_{ij}}dyV(y)/(\zeta \hbar)]$.
Conversely, the complete infinite family of PHJ can be recast in the
EOF provided~(\ref{eq:preV}) is satisfied. This makes it possible to
identify the class of PHJ that is integrable as we shall see later. 

\textit{Embedding in an external potential}.\textit{---} 
The embedding of a system described by EOF in an external potential is known in
the case of a harmonic trap~\citep{Polychronakos92}. For the embedding of a homogenous system in an arbitrary
trapping potential, we draw inspiration from
supersymmetric quantum mechanics~\citep{Cooper95} and introduce
the one-body superpotential $W_{i} \equiv W(x_{i})$ in terms of which 
 the external trapping potential $U_{i}$ will be identified. We 
define the operators
\begin{equation}
a_{i}=\frac{\pi_{i}}{\sqrt{2m}}-\text{i}W_{i},\quad a_{i}^{\dagger}=\frac{\pi_{i}}{\sqrt{2m}}+\text{i}W_{i},
\end{equation}
 and the permutation-invariant quantities $\tilde{I}_{n}\equiv\sum_{i}h_{i}^{n}$,
where $h_{i}\equiv a_{i}^{\dagger}a_{i}$. Projecting $\tilde{I}_{1}$
onto $\mathcal{H}_{\zeta}$, we find $\mathcal{P}_{\zeta}\tilde{I}_{1}\mathcal{P}_{\zeta}=\mathcal{P}_{\zeta}H\mathcal{P}_{\zeta}$,
where $H$ is the Hamiltonian of the system in the presence of the
trap, i.e.,  
\begin{equation}
H=H_{0}+\sum_{i}U_{i}-\zeta\sqrt{\frac{2}{m}}\sum_{i<j}V_{ij}(W_{i}-W_{j}),\label{eq:EOFH-embed}
\end{equation}
 with the external potential $U_{i}$ being determined by the Riccati equation
\begin{equation}
U_{i}=W_{i}^{2}-\frac{\hbar}{\sqrt{2m}}W_{i}'.\label{eq:U-EOF}
\end{equation}
As a familiar example, when $H_{0}$ is the homogeneous Calogero model with inverse-square interactions \cite{Calogero71} and $W_{i}=\sqrt{m/2}\omega x_{i}$,
Eq.~(\ref{eq:EOFH-embed}) reduces to the rational Calogero-Sutherland model \cite{Sutherland71,Sutherland71b} including  a harmonic trap. 

In PHJ, the ground-state wave functions is not limited to the homogeneous
form (\ref{eq:Phi0free}), but also includes more general ground-states
\begin{equation}
\Psi_{0}=\prod_{i<j}f_{ij}\prod_{i}\exp(v_{i})=\Phi_{0}\prod_{i}\exp(v_{i}),\label{eq:Psi0-PHJ}
\end{equation}
where the one-body function $v_{i}=v(x_{i})$ accounts for the role
of the external potential $U_{i}=U(x_{i})$ that breaks translational
invariance \citep{delcampo20}. Specifically, if $H_{0}$ is the parent
hamiltonian of $\Phi_{0}$ in Eq.~(\ref{eq:Phi0free}), then $\Psi_{0}$
has the parent Hamiltonian 
\begin{equation}
H=H_{0}+\sum_{i}U_{i}+\frac{\hbar^{2}}{m}\sum_{i<j}(v_{i}'-v_{j}')\frac{f'_{ij}}{f_{ij}},\label{eq:PHJH-embed}
\end{equation}
with the one-body local external potential $U_{i}$ given in terms
of the function $v_{i}$ by 
\begin{equation}
U_{i}=\frac{\hbar^{2}}{2m}\left[(v_{i}')^{2}+v_{i}''\right].\label{eq:U-PHJ}
\end{equation}
As a result, the Hamiltonian $H$ includes the external potential
$U_{i}$ and an additional pairwise (two-body) potential which is
generally of long-range character.

The equivalence between EOF and PHJ require  that the one-body and potential and the additional long-range term
are equal in both representations.  
Comparing Eq.~(\ref{eq:U-EOF})
and Eq.~(\ref{eq:U-PHJ}), the superpotential $W_{i}$ and the function
$v_{i}$ entering the one-body function of the Jastrow form are related
by
\begin{equation}
W_{i}=-\frac{\hbar}{\sqrt{2m}}v_{i}'.\label{eq:1body-correspond}
\end{equation}
Upon substituting Eq.~(\ref{eq:1body-correspond}) into Eqs.~(\ref{eq:EOFH-embed},~\ref{eq:PHJH-embed}),
we find that the additional long-range potentials coincide, given the correspondence
Eq.~(\ref{eq:preV}) is identified. The ground state of the Hamiltonian
with the external potential in terms of the prepotential and the superpotential
is
\begin{equation}
\Psi_{0}=\exp\left(-\frac{\sqrt{2m}}{\hbar}\sum_{i}\int^{x_{i}}W(y)dy\right)\prod_{i<j}\exp\left[\frac{\int^{x_{ij}}dyV(y)}{\zeta \hbar}\right].\label{eq:Psi0-EOF}
\end{equation}
This establishes the equivalence between EOF and PHJ
in the presence of external potential. 

\textit{Integrability via projection formalism.---} 
For quantum systems with classical analog, as the PHJ, one can define quantum integrability by promoting the Poisson bracket into commutators in the
definition of classical integrability. Polychronakos~\citep{Polychronakos92}
pursued along this line and showed that $I_{n}\equiv\sum_{i}\pi_{i}^{n}$
become integrals of motion, i.e., $[I_{n},\,I_{m}]=0,\,\forall n,\,m$,
in the restricted case in which $V_{ijk}$ vanishes or is constant. 
Having shown that
any spinless model described by EOF is a PHJ with a Jastrow ground state, we next establish the
integrability of the complete family of PHJ models, i.e., without restrictions on the three-body potential $V_{ijk}$ or the external potential $U_i$.

Note that any physical observable $\mathcal{O}$ for spinless
indistingusishable particles must be permutation invariant, i.e.,
$\mathcal{O}(x_{1},\,x_{2},\,\cdots,\,x_{N})=\mathcal{O}(x_{\sigma_{1}},\,x_{\sigma_{2}},\cdots,\,x_{\sigma_{N}})$,$\,\forall$
permutation $\sigma$. As a consequence,
\begin{equation}
[\mathcal{P}_{\zeta},\,\mathcal{O}]=0,\label{eq:PO-comm}
\end{equation}
which can be easily checked  by acting on any wave function in
$\mathcal{H}^{\otimes N}$~\citep{SM}. Eq.~(\ref{eq:PO-comm})
implies the a permutation-invariant observable is block diagonal on
$\mathcal{H}_{\zeta}$ and its orthogonal complement. We define an observable is local if it only involves derivatives with respect to the coordinates up to a finite order. Then permutation invariance
and locality implies that if a permutation-invariant and local observable $\mathcal{O}$
vanishes on $\mathcal{H}_{\zeta}$, then it also vanishes on the full
Hilbert space $\mathcal{H}^{\otimes N}$. That is~\footnote{This observation was put forward first by Polychronakos~\citep{Polychronakos92}, using a heuristic locality argument. However, we emphasize the  argument misses an important ingredient, the permutation invariance. In~\citep{SM}, we give a formal proof inspired by the coordinate Bethe ansatz~\citep{Gaudin14,SamajBajnok13}.},

\begin{equation}
\mathcal{P}_{\zeta}\mathcal{O}\mathcal{P}_{\zeta}=\mathcal{P}_{\zeta}\mathcal{O}=\mathcal{O}\mathcal{P}_{\zeta}=0\Longleftrightarrow\mathcal{O}=0,\label{eq:vanishing-theorem}
\end{equation}
for a permutation-invariant and local observable $\mathcal{O}$.

Eqs.~(\ref{eq:PO-comm},~\ref{eq:vanishing-theorem}) lead to the
following theorem regarding the commutators of two permutation-invariant
observables, which is extremely useful in proving integrability.
\begin{thm}
\label{thm:perm-comm}For two permutation-invariant and local observable $\mathcal{O}_{n}$
and $\mathcal{O}_{m}$, the following three conditions are equivalent
to each other (i) $[\mathcal{O}_{n},\,\mathcal{O}_{m}]=0$ (ii) $\mathcal{P}_{\zeta}[\mathcal{O}_{n},\,\mathcal{O}_{m}]\mathcal{P}_{\zeta}=0$
(iii) $[\mathcal{P}_{\zeta}\mathcal{O}_{n}\mathcal{P}_{\zeta},\,\mathcal{P}_{\zeta}\mathcal{O}_{m}\mathcal{P}_{\zeta}]=0$.
\end{thm}
The equivalence between (i) and (ii) is a consequence of Eq.~(\ref{eq:vanishing-theorem}).
The equivalence between (ii) and (iii) follows from
\begin{align*}
\mathcal{P}_{\zeta}[\mathcal{O}_{n},\,\mathcal{O}_{m}]\mathcal{P}_{\zeta} & =\mathcal{P}_{\zeta}\mathcal{O}_{n}\mathcal{O}_{m}\mathcal{P}_{\zeta}^{2}-\mathcal{P}_{\zeta}\mathcal{O}_{m}\mathcal{O}_{n}\mathcal{P}_{\zeta}^{2}\\
 & =\mathcal{P}_{\zeta}\mathcal{O}_{n}\mathcal{P}_{\zeta}\mathcal{O}_{m}\mathcal{P}_{\zeta}-\mathcal{P}_{\zeta}\mathcal{O}_{m}\mathcal{P}_{\zeta}\mathcal{O}_{n}\mathcal{P}_{\zeta}\\
 & =[\mathcal{P}_{\zeta}\mathcal{O}_{n}\mathcal{P}_{\zeta},\,\mathcal{P}_{\zeta}\mathcal{O}_{m}\mathcal{P}_{\zeta}],
\end{align*}
where we have used Eq.~(\ref{eq:PO-comm}).
\begin{thm}
\label{thm:quantum-integrability}Both the quantum mechanical homogenous
model~(\ref{eq:EOFH}) and the inhomogeneous model~(\ref{eq:EOFH-embed})
generated in EOF are integrable, with the integral of motion being
$I_{n}$ for the homogenous model and $\tilde{I}_{n}$ for the inhomogeneous
model.
\end{thm}
To prove Theorem~\ref{thm:quantum-integrability}, let us first observe
a very interesting property due to the projection $\mathcal{P}_{\zeta}$
and the exchange operator $M_{ij}$. Although
the generalized momentum $\pi_{i}$ involves $N$ degrees of freedom\textit{
}due to the prepotential term, when it is multiplied by $M_{ij}$
from the left, it still satisfies the exchange rule for $1$-body operators~\cite{SM}. As a
consequence,
\begin{equation}
\mathcal{P}_{\zeta}\pi_{i}^{n}\mathcal{P}_{\zeta}=\mathcal{P}_{\zeta}M_{ij}^{2}\pi_{i}^{n}\mathcal{P}_{\zeta}=\mathcal{P}_{\zeta}M_{ij}\pi_{i}^{n}M_{ij}\mathcal{P}_{\zeta}=\mathcal{P}_{\zeta}\pi_{j}^{n}\mathcal{P}_{\zeta}.\label{eq:super-symmetic-1body}
\end{equation}
A similar equation also holds for $h_{i}$. For a more general version
of this identity, see~\citep{SM}. 

On the other hand, since the integrals of motions $I_{n}$ and $\tilde{I}_{n}$
are permutation invariant and local, one can reduce their commutativity
to condition (iii) in Theorem~\ref{thm:perm-comm}. Using Eq.~(\ref{eq:super-symmetic-1body})
or the analogous equation for $h_{i}$, it follows that the
condition (iii) in Theorem~\ref{thm:perm-comm} is satisfied, with $\mathcal{O}_{n}=I_{n}$
or $\mathcal{O}_{n}=\tilde{I}_{n}$. This concludes the proof of the
integrability of the Hamiltonians~(\ref{eq:EOFH}) and~(\ref{eq:EOFH-embed}).

A few comments are in order. First, since we have proved the equivalence
between EOF and PHJ, the Hamiltonians~(\ref{eq:PHJH}) and~(\ref{eq:PHJH-embed})
are therefore also integrable.
Second, it is possible to build integrals of motions for families
of classical models generated by the EOF and PHJ according to the
quantum-classical correspondence. One can expand the powers in $I_{n}$
and $\tilde{I}_{n}$ and compute $\mathcal{P}_{\zeta}I_{n}\mathcal{P}_{\zeta}$
and $\mathcal{P}_{\zeta}\tilde{I}_{n}\mathcal{P}_{\zeta}$ explicitly
with Eq.~(\ref{eq:MP-PM-P}). Then one is left with the expressions $\mathcal{P}_{\zeta}K_{n}\mathcal{P}_{\zeta}$
and $\mathcal{P}_{\zeta}\tilde{K}_{n}\mathcal{P}_{\zeta}$ , where
$K_{n}$ and $\tilde{K}_{n}$ contain only the phase space variables
but no exchange operators. In particular, we note that $K_{2}=H_{0}$
and $\tilde{K}_{1}=H$; see ~\citep{SM},
where one obtains $H_{0}$ by projecting $I_{2}$ onto $\mathcal{H}_{\zeta}$.
According to Theorem~\ref{thm:perm-comm}, $K_{n}$'s and $\tilde{K}_{n}$'s
must commute on the whole Hilbert space $\mathcal{H}^{\otimes N}$,
respectively. Transitioning to the classical model, where the commutator
is demoted to Poisson brackets, the Poisson brackets of $K_{n}$'s
and $\tilde{K}_{n}$'s must vanish, respectively. Thus, we see that
$K_{n}$'s and $\tilde{K}_{n}$'s are also the integrals of motions
for the classical model with Hamiltonians~(\ref{eq:EOFH},~\ref{eq:PHJH})
and (\ref{eq:EOFH-embed},~\ref{eq:PHJH-embed}), respectively.

{
\textit{Discussion.---} It is worth noting that the Jastrow wave
functions $\Phi_{0}$, $\Psi_{0}$ may not be the true ground state
of the corresponding PHJ if they cannot be properly normalized.
Nevertheless, the family of models generated
by EOF and PHJ is always integrable, regardless of the normalization
of the Jastrow wave function.

For example, if $f_{ij}=\exp(g|x_{ij}|)$, $H_{0}$ becomes the well-known
LL model~\citep{delcampo20}. However, the Jastrow wave function
is normalizable only when $g<0$, which corresponds to the McGuire
bright soliton~\citep{McGuire64}. Therefore $\Phi_{0}$ is no longer
the ground state wave function of the repulsive LL model. However,
as we have discussed previously, the integrability of the Hamiltonian
is not affected, so our result reproduces the integrability of the
LL model with the integral of motion being $I_{n}$ or $K_{n}$. More
interestingly, upon introducing the external harmonic potential, according
to Eq.~(\ref{eq:Psi0-PHJ}), $\Psi_{0}$ becomes normalizable even
if $\Phi_{0}$ is not and Eq.~(\ref{eq:PHJH}) corresponds to the Lieb-Liniger-Coulomb
model introduced in Refs.~\citep{BeauPittman20}, i.e.,
\begin{equation}
H=\sum_{i}\left[\frac{p_{i}^{2}}{2m}+\frac{1}{2}m\omega^{2}x_{i}^{2}\right]+g\sum_{i<j}\left[\frac{2\hbar^{2}}{m}\delta(x_{ij})-\frac{m\omega}{\hbar}|x_{ij}|\right],\label{eq:LLC}
\end{equation}
with ground state $E_0=\frac{N\hbar\omega}{2}-\frac{g^2\hbar^2}{m}\frac{N(N^2-1)}{6}$. This system  describes harmonically confined bosons subject to contact and Coulomb interactions or gravitational attraction in one spatial dimension.
Ref.~\citep{BeauPittman20} characterized its EOF representation and ground state properties. Using Theorem~\ref{thm:quantum-integrability}, we conclude that this system is integrable, with the integrals of motion being $\tilde{I}_{n}\equiv\sum_{i}h_{i}^{n}$. 

Further physical examples of integrable PHJ systems are provided in the Supplemental Material~\cite{SM}, which includes Refs. \cite{Calogero75,Beau17note}. The proof leading to the integrability of PHJ essentially takes advantage of the permutation invariance and EOF. As a result it can be applied to models defined on the real line as well as those embedded in an external potential. Likewise, it  holds for  systems with hard-wall confinement or  a ring geometry, provided the  pair function $f_{ij}$ and the one body potential $v_i$ or $W_i$  fulfill the corresponding boundary conditions.

\textit{Conclusion.---} We have established the equivalence
between the families of one-dimensional many-body quantum systems
generated by the exchange operator formalism and parent Hamiltonians with a ground-state wavefunction of Jastrow form, describing indistinguishable particles with
no internal degrees of freedom. 
 Making use of
the projection operator onto the spinless bosonic or fermionic subspace,
we have proved the integrability of all these systems
by constructing explicitly 
the corresponding integrals of motion.
Embedding these translation-invariant models in an external
potential preserves the integrability, in the presence of long-range interactions, 
as we have illustrated in the long-range Lieb-Liniger model and related systems. 

These findings advance the study of many-body physics by uncovering the implications of ground-state correlations on integrability.  They
 should lead to manifold applications in
 the study of quantum solitons, quantum quenches and the thermalization of isolated integrable systems (governed by integrals of motion), and strongly-correlated regimes, generalizing the super-Tonks-Girardeau gas \cite{STG1}, among others.}
  Our results bear also implications on numerical methods for strongly-correlated systems such as variational methods and quantum Monte Carlo algorithms, in which the ubiquitous use of Jastrow trial wavefunctions may impose integrability on systems lacking it.
It may be possible to extend our results to higher spatial dimensions \citep{BeauDC21}, higher-order correlations \cite{Carleo17},
the inclusion of spin degrees of freedom \citep{Polychronakos06},
mixtures of different species \citep{GirardeauMinguzzi07},  and distinguishable
particles \citep{JainKhare99,Pittman17}.

\textit{Acknowledgement.---} It is a pleasure to acknowledge discussions with Pablo Martinez Azcona and Aritra Kundu.

\let\oldaddcontentsline\addcontentsline     
\renewcommand{\addcontentsline}[3]{}         

\bibliographystyle{apsrev4-1}
\bibliography{BAIQS_lib}

\clearpage\newpage\setcounter{equation}{0} \setcounter{section}{0}
\setcounter{subsection}{0} 
\global\long\def\theequation{S\arabic{equation}}%
\onecolumngrid \setcounter{enumiv}{0} 

\setcounter{equation}{0} \setcounter{section}{0} \setcounter{subsection}{0} \renewcommand{\theequation}{S\arabic{equation}} \onecolumngrid \setcounter{enumiv}{0}
\begin{center}
\textbf{\large{}Supplemental Material for \\``One-Dimensional Quantum Systems with Ground-State of Jastrow Form Are Integrable''}{\large\par}
\par\end{center}

\let\addcontentsline\oldaddcontentsline     

\tableofcontents

\addtocontents{toc}{\protect\thispagestyle{empty}}
\pagenumbering{gobble}

\section{The exchange operator}

In this section, we give a rigorous definition of the exchange operator
used in the main text and derive its properties from first principles.

\subsection{Properties of the two-particle position exchange operator $M_{ij}$}

Consider the $N$-fold tensor-product of a single-particle Hilbert space $\mathcal{H}$, i.e., $\mathcal{H}^{\otimes N}$,
where we do not assume a particular exchange statistics for the particles. 
This means that the particles may be distinguishable or
not. The position exchange operator $M_{ij}$ acting on the spatial coordinates $x_{i}$
and $x_{j}$ on a many-particle state $\psi_{s}\in\mathcal{H}^{\otimes N}$ is defined as 
\begin{equation}
M_{ij}\psi_{s}(x_{1},\cdots x_{i},\,\cdots,\,x_{j},\,\cdots x_{N})=\psi_{s}(x_{1},\cdots x_{j},\,\cdots,\,x_{i},\,\cdots x_{N}),\,i\neq j,\label{eq:M-def}
\end{equation}
where $\psi_{s}(x_{1},\cdots x_{i},\,\cdots,\,x_{j},\,\cdots x_{N})$
is the wave function on $\mathcal{H}^{\otimes N}$ and $s$ denotes
the internal degrees of freedom, which can account for spin in identical
particles. By definition, we find 
\begin{equation}
M_{ji}=M_{ij}.\label{eq:M-sym}
\end{equation}
Furthermore, according to this definition, it follows that
\begin{equation}
M_{ij}^{2}\psi_{s}(x_{1},\cdots x_{i},\,\cdots,\,x_{j},\,\cdots x_{N})=\psi_{s}(x_{1},\cdots x_{i},\,\cdots,\,x_{j},\,\cdots x_{N}),
\end{equation}
which implies that $M_{ij}$ is idempotent, i.e., 
\begin{equation}
M_{ij}^{2}=\mathbb{I}.\label{eq:M2}
\end{equation}
The inner product on $\mathcal{H}^{\otimes N}$ is 
\begin{equation}
(\phi,\,\psi)\equiv\sum_{s}\int\prod_{k=1}^{N}dx_{k}\phi_{s}^{*}(x_{1},\cdots x_{i},\,\cdots,\,x_{j},\,\cdots x_{N})\psi_{s}(x_{1},\cdots x_{i},\,\cdots,\,x_{j},\,\cdots x_{N}).\label{eq:in-prod}
\end{equation}
The adjoint of $M_{ij}$ is defined as $(M_{ij}^{\dagger}\phi,\,\psi)\equiv(\phi,\,M_{ij}\psi)$.
Explicitly expanding the inner product according to Eq~(\ref{eq:in-prod}),
we arrive at
\begin{align}
 & \sum_{s}\int\prod_{k=1}^{N}dx_{k}[M_{ij}^{\dagger}\phi_{s}(x_{1},\cdots,\,x_{i},\,\cdots x_{j},\,\cdots x_{N})]^{*}\psi_{s}(x_{1},\cdots x_{i},\,\cdots,\,x_{j},\,\cdots x_{N})\nonumber \\
= & \sum_{s}\int\prod_{k=1}^{N}dx_{k}\phi_{s}^{*}(x_{1},\cdots,\,x_{i},\,\cdots x_{j},\,\cdots x_{N})\psi_{s}(x_{1},\cdots x_{j},\,\cdots,\,x_{i},\,\cdots x_{N})\nonumber \\
= & \sum_{s}\int\prod_{k=1}^{N}dx_{k}\phi_{s}^{*}(x_{1},\cdots,\,x_{j},\,\cdots x_{i},\,\cdots x_{N})M_{ij}\psi_{s}(x_{1},\cdots x_{i},\,\cdots,\,x_{j},\,\cdots x_{N}),
\end{align}
which indicates that 
\begin{equation}
M_{ij}^{\dagger}\phi_{s}(x_{1},\cdots x_{i},\,\cdots,\,x_{j},\,\cdots x_{N})=\phi_{s}(x_{1},\cdots x_{j},\,\cdots,\,x_{i},\,\cdots x_{N}).
\end{equation}
Comparing above equation with Eq.~(\ref{eq:M-def}), we conclude that $M_{ij}$ is a Hermitian operator, 
\begin{equation}
M_{ij}=M_{ij}^{\dagger}.
\end{equation}
We define operators or observables 
\begin{eqnarray}
A_{i} & \equiv&A(x_{i}),\\
A_{ij} & \equiv&A(x_{i},\,x_{j}),\\
A_{ijk} & \equiv&A(x_{i},\,x_{j},\,x_{k}),\\
&\vdots & \nonumber 
\end{eqnarray}
where the dependence on the internal degrees of freedom of $\mathcal{O}$
is suppressed. Then one can easily find that
\begin{equation}
(\phi,\,M_{ij}A_{k}\psi)=(M_{ij}^{\dagger}\phi,\,A_{k}\psi)=\sum_{s}\int\prod_{l=1}^{N}dx_{l}\phi_{\sigma}^{*}(x_{1},\cdots,\,x_{j},\,\cdots x_{i},\,\cdots x_{N})\mathcal{O}(x_{k})\psi_{s}(x_{1},\cdots x_{i},\,\cdots,\,x_{j},\,\cdots x_{N})
\end{equation}
For $k\neq i,\,j$,  
\begin{align}
 & \sum_{s}\int\prod_{l=1}^{N}dx_{l}\phi_{s}^{*}(x_{1},\cdots,\,x_{j},\,\cdots x_{i},\,\cdots x_{N})\mathcal{O}(x_{k})\psi_{s}(x_{1},\cdots x_{i},\,\cdots,\,x_{j},\,\cdots x_{N})\nonumber \\
= & \sum_{s}\int\prod_{l=1}^{N}dx_{l}\phi_{s}^{*}(x_{1},\cdots,\,x_{j},\,\cdots x_{i},\,\cdots x_{N})\mathcal{O}(x_{k})\psi_{s}(x_{1},\cdots x_{i},\,\cdots,\,x_{j},\,\cdots x_{N})\nonumber \\
= & \sum_{s}\int\prod_{l=1}^{N}dx_{l}\phi_{s}^{*}(x_{1},\cdots,\,x_{j},\,\cdots x_{i},\,\cdots x_{N})\mathcal{O}(x_{k})M_{ji}\psi_{s}(x_{1},\cdots x_{j},\,\cdots,\,x_{i},\,\cdots x_{N})\nonumber \\
= & (\phi,\,\mathcal{O}_{k}M_{ji})
\end{align}
We conclude that 
\begin{equation}
M_{ij}A_{k}=A_{k}M_{ji}=A_{k}M_{ij},\quad \text{for }\,i,\,j,\,k,\,\text{distinct.}\label{eq:MO1}
\end{equation}
If $k=i$, then we find 
\begin{align}
(\phi,\,M_{ij}A_{i}\psi) & =\sum_{s}\int\prod_{l=1}^{N}dx_{l}\phi_{s}^{*}(x_{1},\cdots,\,x_{j},\,\cdots x_{i},\,\cdots x_{N})\mathcal{O}(x_{i})\psi_{s}(x_{1},\cdots x_{i},\,\cdots,\,x_{j},\,\cdots x_{N})\nonumber \\
 & =\sum_{s}\int\prod_{l=1}^{N}dx_{l}\phi_{s}^{*}(x_{1},\cdots,\,x_{j},\,\cdots x_{i},\,\cdots x_{N})A(x_{i})M_{ji}\psi_{s}(x_{1},\cdots x_{j},\,\cdots,\,x_{i},\,\cdots x_{N})\nonumber \\
 & =\sum_{s}\int\prod_{l=1}^{N}dx_{l}\phi_{s}^{*}(x_{1},\cdots,\,x_{i},\,\cdots x_{j},\,\cdots x_{N})A(x_{j})M_{ij}\psi_{s}(x_{1},\cdots x_{i},\,\cdots,\,x_{j},\,\cdots x_{N})\nonumber \\
 & =(\phi,\,A_{j}M_{ij}\psi),
\end{align}
where we have used change of the dummy indices $x_{i}\to x_{j}$ and
$x_{j}\to x_{i}$ in the second last equation. We find 
\begin{equation}
M_{ij}A_{i}=A_{j}M_{ij}
\end{equation}
and similarly 
\begin{equation}
M_{ij}A_{j}=A_{i}M_{ij}.
\end{equation}
One can extend the above arguments to many-body operators $A_{ijk\cdots}$
with little effort. For example, for two-body operators $A_{kl}$,
one can find 
\begin{align}
M_{ij}A_{kl} & =A_{kl},\quad \text{for }\,i,\,j,\,k,\,l,\,\text{distinct}\\
M_{ij}A_{jk} & =A_{ik}M_{ij},\quad \text{for }\,i,\,j,\,k,\,\text{distinct}\\
M_{ij}A_{kj} & =A_{ki}M_{ij},\quad \text{for }\,i,\,j,\,k,\,\text{distinct}\\
M_{ij}A_{ij} & =A_{ji}M_{ij},\\
M_{ij}A_{ji} & =A_{ij}M_{ij}.\label{eq:MO2}
\end{align}

\subsection{Properties of the three-particle position exchange operators $M_{ijk}$}

The particle exchange operator with three indices is defined as 
\begin{equation}
M_{ijk}=M_{ij}M_{jk},\quad \text{for }\\,i,\,j,\,k,\,\text{distinct}.\label{eq:3perm-def}
\end{equation}
With the properties of the two-particle exchange operator, one can
easily show that $M_{ijk}$ is invariant under cyclic permutation,
i.e.,
\begin{equation}
M_{ijk}=M_{jki}=M_{kij}
\end{equation}
However, it is {\it not} fully symmetric in its indices. In particular, $M_{ijk}\neq M_{jik}$. Furthemore, for all distinct $i,\,j,\,k,\,l$,
\begin{equation}
M_{ijk}A_{l}=A_{l}M_{ijk}
\end{equation}
since $A_{l}$ commute with $M_{ij}$ and $M_{jk}$. Furthermore,
one can explicit check that 
\begin{align}
M_{ijk}A_{i} & =M_{ij}A_{i}M_{jk}=A_{j}M_{ijk},\label{eq:3perm-sym1}\\
M_{ijk}A_{j} & =M_{ij}A_{k}M_{jk}=A_{k}M_{ijk},\\
M_{ijk}A_{k} & =M_{ij}A_{j}M_{jk}=A_{i}M_{ijk}.\label{eq:3perm-sym3}
\end{align}

\section{\label{sec:Caveat-for-projection} Projection onto the bosonic and fermionic
subspaces}

In EOF, one may restrict the construction to the bosonic or fermionic subspace to simplify
the calculation or motivated on physical grounds. For example, it was shown by Polychronakos~\citep{Polychronakos92,Polychronakos06}
that if the prepotential is $V_{ij}=l/x_{ij}$, then 
\begin{equation}
\frac{1}{2}I_{2}=\frac{1}{2m}\sum_{i}\pi_{i}^{2}=\frac{1}{2m}\sum_{i}p_{i}^{2}+\frac{\hbar^{2}}{m}\sum_{i>j}\frac{l(l-M_{ij})}{(x_{i}-x_{j})^{2}}.\label{eq:half-I2}
\end{equation}
This Hamiltonian is still defined on $\mathcal{H}^{\otimes N}$. Recall
that the Hilbert space $\mathcal{H}_{\zeta}$ of identical particles
is a subspace of $\mathcal{H}^{\otimes N}$. We define the projection
operator on $\mathcal{H}^{\otimes N}$ to $\mathcal{H}_{\zeta}$ as
$\mathcal{P}_{\zeta}$. It can be readily checked that for any $\psi(x_{1},\,x_{2},\,\cdots,\,x_{N})\in\mathcal{H}^{\otimes N}$,
\begin{align}
M_{ij}\mathcal{P}_{\zeta}\psi(x_{1},\,x_{2},\,\cdots,\,x_{N}) & =\frac{1}{N!}\sum_{\sigma}(\pm1)^{\sigma}M_{ij}\psi(x_{\sigma_{1}},\dots\,x_{\sigma_{i}},\cdots x_{\sigma_{j}}\,\cdots,\,x_{\sigma_{N}})\nonumber \\
 & =\frac{1}{N!}\sum_{\sigma}(\pm1)^{\sigma+1}\psi(x_{\sigma_{1}},\dots\,x_{\sigma_{i}},\cdots x_{\sigma_{j}}\,\cdots,\,x_{\sigma_{N}})\nonumber \\
 & =\pm\mathcal{P}_{\zeta}\psi(x_{1},\,x_{2},\,\cdots,\,x_{N}),
\end{align}
\begin{align}
\mathcal{P}_{\zeta}M_{ij}\psi(x_{1},\,x_{2},\,\cdots,\,x_{N}) & =\mathcal{P}_{\zeta}\psi(x_{1},\,\cdots x_{j},\,\cdots x_{i},\cdots\,x_{N})\nonumber \\
 & =(\pm1)\mathcal{P}_{\zeta}\psi(x_{1},\,\cdots x_{i},\,\cdots x_{j},\cdots\,x_{N})\nonumber \\
 & =\frac{1}{N!}\sum_{\sigma}(\pm1)^{\sigma+1}\psi(x_{\sigma_{1}},\dots\,x_{\sigma_{i}},\cdots x_{\sigma_{j}}\,\cdots,\,x_{\sigma_{N}})\nonumber \\
 & =\pm\mathcal{P}_{\zeta}\psi(x_{1},\,x_{2},\,\cdots,\,x_{N}).
\end{align}
Thus, we conclude
\begin{equation}
[\mathcal{P}_{\zeta},\,M_{ij}]=0,
\end{equation}
\begin{equation}
\mathcal{P}_{\zeta}M_{ij}=M_{ij}\mathcal{P}_{\zeta}=\zeta\mathcal{P}_{\zeta}.
\end{equation}
Therefore, when projecting onto the bosonic or fermionic subspace, Eq.~(\ref{eq:half-I2})
becomes
\begin{equation}
\frac{1}{2}\mathcal{P}_{\zeta}I_{2}\mathcal{P}_{\zeta}=\frac{1}{2}\sum_{i}\pi_{i}^{2}=\mathcal{P}_{\zeta}H_{\text{CS}}\mathcal{P}_{\zeta},\label{eq:PHP}
\end{equation}
where 
\begin{equation}
H_{\text{CS}}=\frac{1}{2m}\sum_{i}p_{i}^{2}+\frac{\hbar^{2}}{m}\sum_{i>j}\frac{l(l\mp\zeta)}{(x_{i}-x_{j})^{2}}
\end{equation}
 is the rational Calogero model. The procedure in Eq.~(\ref{eq:PHP})
is effectively equivalent to replacing $M_{ij}=\zeta\mathbb{I}$. 

However, we warn the audience that when focusing on spinless bosons
or fermions, setting $M_{ij}=\zeta\mathbb{I}$ requires some caution.
For example, when calculating $[p_{i},\,\sum_{k\neq j}V_{jk}M_{jk}]$
with $i\neq j$, had one setting $M_{jk}=\zeta\mathbb{I}$ before
actually calculating the commutator, one would obtain 
\begin{equation}
[p_{i},\,\sum_{k\neq j}V_{jk}M_{jk}\mathcal{P}_{\zeta}]\overset{?}{=}\zeta\text{i}\hbar V_{ij}^{\prime}.
\end{equation}
On the other hand, a rigorous calculation following Eqs.~(\ref{eq:MO1}-\ref{eq:MO2})
shows that 
\begin{align}
[p_{i},\,\sum_{k\neq j}V_{jk}M_{jk}] & =[p_{j},\,V_{ji}M_{ji}]=[p_{i},\,V_{ji}]M_{ji}+V_{ji}[p_{i},\,M_{ji}]\nonumber \\
 & =\text{i}\hbar V_{ij}^{\prime}M_{ij}-V_{ij}(p_{i}-p_{j})M_{ij},\,(i\neq j)
\end{align}
Upon projecting this relation onto $\mathcal{H}_{\alpha}$ by setting $M_{jk}=\mathbb{I}$,
which yields 
\begin{equation}
[p_{i},\,\sum_{k\neq j}V_{jk}M_{jk}]\mathcal{P}_{\zeta}=\zeta\text{i}\hbar V_{ij}^{\prime}-\zeta V_{ij}(p_{i}-p_{j}).
\end{equation}
This example clearly shows that $[\mathcal{A},\,\mathcal{B}]\mathcal{P}_{\zeta}\neq[\mathcal{A},\,\mathcal{B}\mathcal{P}_{\zeta}]$
and in particular 
\begin{equation}
[\mathcal{A},\,\mathcal{B}\mathcal{P}_{\zeta}]=[\mathcal{A},\,\mathcal{B}]\mathcal{P}_{\zeta}+\mathcal{B}[\mathcal{A},\,\mathcal{P}_{\zeta}].
\end{equation}
The caveat we would like to give to the audience is that whenever
one would like to set $M_{ij}=\pm\mathbb{I}$, say, on physical grounds, one should
bear in mind that at the formal level a projection has been introduced, which has to be taken into account to work out the correct algebra. 

Finally, we mention the following interesting lemma for spinless indistinguishable
particles, thanks to the EOF and the projection operator:
\begin{lem}
\label{lem:super-symmetric} For any $k$-body operator $A_{i_{1}\cdots i_{k}}$
acting on states of $N$-spinless indistinguishable particles, $\mathcal{P}_{\zeta}A_{i_{1}\cdots i_{k}}\mathcal{P}_{\zeta}$
is ``super-symmetric'', i.e., 
\begin{equation}
\mathcal{P}_{\zeta}A_{i_{1}\cdots i_{k}}\mathcal{P}_{\zeta}=\mathcal{P}_{\zeta}A_{j_{1}\cdots j_{k}}\mathcal{P}_{\zeta},
\end{equation}
where $(j_{1},\,j_{2},\,\cdots,\,j_{k})$ is arbitrary $k$-tuple,
not necessarily a permutation of $(i_{1},\,\cdots i_{k})$
\end{lem}
\begin{proof}
The proof is to take advantage of Eq.~(\ref{eq:M2}) and Eq.~(\ref{eq:MP-PM-P}).
The intuition for the proof can be easily seen from the case where indices
$(j_{1},\,j_{2},\,\cdots,\,j_{k})$ are distinct:
\begin{align}
\mathcal{P}_{\zeta}A_{i_{1}\cdots i_{k}}\mathcal{P}_{\zeta} & =\mathcal{P}_{\zeta}M_{i_{1}j_{1}}\cdots M_{i_{k}j_{k}}M_{i_{k}j_{k}}\cdots M_{i_{1}j_{1}}A_{i_{1}\cdots i_{k}}\mathcal{P}_{\zeta}\nonumber \\
 & =\mathcal{P}_{\zeta}M_{i_{1}j_{1}}\cdots M_{i_{k}j_{k}}A_{j_{1}\cdots j_{k}}M_{i_{k}j_{k}}M_{i_{k}j_{k}}M_{i_{1}j_{1}}\mathcal{P}_{\zeta}\nonumber \\
 & =\mathcal{P}_{\zeta}\zeta^{k}A_{j_{1}\cdots j_{k}}\zeta^{k}\mathcal{P}_{\zeta}\nonumber \\
 & =\mathcal{P}_{\zeta}A_{j_{1}\cdots j_{k}}\mathcal{P}_{\zeta}.
\end{align}
With the above intuition, one can easily prove the general case without
any difficulty.
\end{proof}

\section{\label{sec:Derivation-of-EOFH0}Derivation of Eq.~(\ref{eq:EOFH})}

Given the generalized
momenta 
\[
\pi_{i}=p_{i}+\text{i}\sum_{j\neq i}V_{ij}M_{ij},
\]
the EOF Hamiltonian is defined as 
\begin{eqnarray}
\frac{1}{2m}I_{2} & = & \frac{1}{2m}\sum_{i}\pi_{i}^{2}\nonumber \\
 & = & \frac{1}{2m}\sum_{i}\left(p_{i}+\text{i}\sum_{j\neq i}V_{ij}M_{ij}\right)^{2}\nonumber \\
 & = & \frac{1}{2m}\sum_{i}\left[p_{i}^{2}+\text{i}\sum_{j\neq i}(p_{i}V_{ij}M_{ij}+V_{ij}M_{ij}p_{i})-\sum_{j,\,k\neq i}V_{ij}M_{ij}V_{ik}M_{ik}\right],\label{eq:I2-explicit}
\end{eqnarray}
which is the sum of a purely kinetic one-body term, a two-body term,
and a three-body term.We next use the fact the two-body term can be
rewritten making use of $\ensuremath{[p_{i},V_{ij}]=-\text{i}\hbar,\,\partial_{i}V_{ij}=-\text{i}\hbar V_{ij}'}$
together with $\ensuremath{V_{ij}M_{ij}p_{i}=V_{ij}p_{j}M_{ij}}$.

In addition, we note that the third term on r.h.s. of Eq.~(\ref{eq:I2-explicit})
with $j=k$ becomes a two-body term $\ensuremath{-V_{ij}M_{ij}V_{ij}M_{ij}=-V_{ij}V_{ji}M_{ij}^{2}=+V_{ij}^{2}}$,
using the fact that $\ensuremath{V_{ij}}$ is antisymmetric in its
indices. Thus 
\begin{align}
\frac{1}{2m}I_{2} & =\frac{1}{2m}\sum_{i}\left[p_{i}^{2}+\sum_{j\neq i}(\hbar V_{ij}'M_{ij}+\text{i}V_{ij}(p_{i}+p_{j})M_{ij}+V_{ij}^{2})-\sum_{i\neq j\neq k\neq i}V_{ij}M_{ij}V_{ik}M_{ik}\right]\nonumber \\
 & =\frac{1}{2m}\sum_{i}\left[p_{i}^{2}+\sum_{j\neq i}(\hbar V_{ij}'M_{ij}+V_{ij}^{2})-\sum_{i\neq j\neq k\neq i}V_{ij}V_{jk}M_{ij}M_{ik}\right]\nonumber \\
 & =\frac{1}{2m}\sum_{i}\left[p_{i}^{2}+\sum_{j\neq i}(\hbar V_{ij}'M_{ij}+V_{ij}^{2})-\sum_{i\neq j\neq k\neq i}V_{ij}V_{ki}M_{ijk}\right],\label{eq:I2-2m}
\end{align}
where in the last step we have swapped the indices $i$ and $j$ and
used the fact that $V_{ij}$ is antisymmetric in its indices, and that $M_{ijk}$
is the three-particle exchange operator defined in Eq.~(\ref{eq:3perm-def}).
Taking advantage of the invariance of $M_{ijk}$ under the cyclic
permutation of $i$, $j$, $k$, we can further symmetrize the last
term on the r.h.s. of Eq.~(\ref{eq:I2-2m}) as
\begin{equation}
\sum_{i\neq j\neq k\neq i}V_{ij}V_{ki}M_{ijk}=\frac{1}{3}\sum_{j\neq k\neq i}(V_{ij}V_{jk}+V_{jk}V_{ki}+V_{ki}V_{ij})M_{ijk}=\frac{1}{3}\sum_{j\neq k\neq i}V_{ijk}M_{ijk},
\end{equation}
where $V_{ijk}=V_{ij}V_{jk}+V_{jk}V_{ki}+V_{ki}V_{ij}$. Upon projecting
to the bosonic or fermionic subspace, we find, 

\begin{equation}
\frac{1}{2m}\mathcal{P}_{\zeta}I_{2}\mathcal{P}_{\zeta}=\mathcal{P}_{\zeta}\left(\frac{1}{2m}\sum_{i}p_{i}^{2}+\frac{1}{2m}\sum_{i\neq j}(\zeta\hbar V_{ij}'+V_{ij}^{2})-\frac{1}{6m}\sum_{j\neq k\neq i}V_{ijk}\right)\mathcal{P}_{\zeta}
\end{equation}
Since $V_{ij}^{\prime}$ and $V_{ij}^{2}$ are symmetric in the indices
$i$ and $j$ and $V_{ijk}$ is fully symmetric in permutation of
any pair of the indices, the sums in above equation can be further
rewritten as 
\begin{equation}
\frac{1}{2m}\mathcal{P}_{\zeta}I_{2}\mathcal{P}_{\zeta}=\mathcal{P}_{\zeta}\left(\frac{1}{2m}\sum_{i}p_{i}^{2}+\frac{1}{m}\sum_{i<j}(\zeta\hbar V_{ij}'+V_{ij}^{2})-\frac{1}{m}\sum_{i<j<k}V_{ijk}\right)\mathcal{P}_{\zeta}.
\end{equation}
Therefore, we find Eq.~(\ref{eq:EOFH}).

\section{Results related to integrability }

In this section, we discuss some results that are used in the proof
of integrability discussed in the main text.

\subsection{Proof of Eq.~(\ref{eq:PO-comm})}

For any permutation-invariant observable $\mathcal{O}(x_{1},\,x_{2},\,\cdots,\,x_{N})$,
\begin{align}
\mathcal{P}_{\zeta}\mathcal{O}(x_{1},\,x_{2},\,\cdots,\,x_{N})\psi(x_{1},\,x_{2},\,\cdots,\,x_{N}) & =\frac{1}{N!}\sum_{\sigma}\mathcal{O}(x_{\sigma_{1}},\,x_{\sigma_{2}},\,\cdots,\,x_{\sigma_{N}})\psi(x_{\sigma_{1}},\,x_{\sigma_{2}},\,\cdots,\,x_{\sigma_{N}})\nonumber \\
 & =\mathcal{O}(x_{1},\,x_{2},\,\cdots,\,x_{N})\mathcal{P}_{\zeta}\psi(x_{1},\,x_{2},\,\cdots,\,x_{N}),
\end{align}
for any $\psi(x_{1},\,x_{2},\,\cdots,\,x_{N})\in\mathcal{H}^{\otimes N}$.
Therefore, we conclude that
\begin{equation}
[\mathcal{P}_{\zeta},\,\mathcal{O}]=0.
\end{equation}
. 

\subsection{Proof of the block-diagonal structure of permutation-invariant observables.}
\begin{prop*}
A permutation-invariant observable is block diagonal on $\mathcal{H}_{\zeta}$
and $\mathcal{H}_{\zeta}^{\perp}$, where $\mathcal{H}_{\zeta}^{\perp}$
is the orthogonal complement of $\mathcal{H}_{\zeta}$ on $\mathcal{H}^{\otimes N}$.
\end{prop*}
\begin{proof}
We denote the projector onto $\mathcal{H}_{\zeta}^{\perp}$ as $\mathcal{P}_{\zeta}^{\perp}$.
Then, for a permutation-invariant observable $\mathcal{O}$, according
to Eq.~(\ref{eq:PO-comm}), it can be easily seen 
\begin{align}
\mathcal{P}_{\zeta}\mathcal{O}\mathcal{P}_{\zeta}^{\perp} & =\mathcal{O}\mathcal{P}_{\zeta}\mathcal{P}_{\zeta}^{\perp}=0,\\
\mathcal{P}_{\zeta}^{\perp}\mathcal{O}\mathcal{P}_{\zeta} & =\mathcal{P}_{\zeta}\mathcal{P}_{\zeta}^{\perp}\mathcal{O}=0,
\end{align}
whence it follows that 
\begin{align}
\mathcal{O}=\mathcal{P}_{\zeta}\mathcal{O}\mathcal{P}_{\zeta}\oplus\mathcal{P}_{\zeta}^{\perp}\mathcal{O}\mathcal{P}_{\zeta}^{\perp},
\end{align}
 which concludes the proof. 
\end{proof}

\subsection{Proof of Eq.~(\ref{eq:vanishing-theorem})}

Eq.~(\ref{eq:vanishing-theorem}) is equivalent to the following
proposition:

\begin{prop*}
If a permutation-invariant and local observable $\mathcal{O}$ vanishes
on $\mathcal{H}_{\zeta}$, then it must vanish on the full Hilbert
space $\mathcal{H}^{\otimes N}$.
\end{prop*}

\begin{proof}
Assuming $\exists\psi(x_{1},\,x_{2},\,\cdots,\,x_{N})\in\mathcal{H}^{\otimes N}$,
which is not necessarily symmetric or anti-symmetric.
We borrow inspiration from coordinate Bethe ansatz~\citep{Gaudin14,SamajBajnok13}
and construct another symmetric or anti-symmetric wave function $\phi(x_{1},\,x_{2},\,\cdots,\,x_{N})\in\mathcal{H}_{\zeta}$
\begin{equation}
\phi(x_{1},\,x_{2},\,\cdots,\,x_{N})\equiv\sum_{\sigma}\zeta^{\sigma}\psi(x_{\sigma_{1}},\,x_{\sigma_{2}},\,\cdots x_{x_{\sigma_{N}}})\theta_{\text{H}}(x_{\sigma_{N}}-x_{\sigma_{N-1}})\cdots\theta_{\text{H}}(x_{\sigma_{3}}-x_{\sigma_{2}})\theta_{\text{H}}(x_{\sigma_{2}}-x_{\sigma_{1}})
\end{equation}
where $\sigma$ is a permutation of the tuple $(1,\,2,\,\cdots,\,N)$
and $\theta_{\text{H}}(x)$ is the Heaviside function. It is worth
to note that $\phi$ is not $\mathcal{P}_{\zeta}\psi$. For example,
if $\psi\in\mathcal{H}_{\zeta}^{\perp}$, then $\mathcal{P}_{\zeta}\psi=0$
while $\phi\neq0$ as long as $\psi\neq0$. Since $\mathcal{O}$ vanishes
on $\mathcal{H}_{\zeta}$ therefore 
\begin{equation}
\mathcal{O}(x_{1},\,x_{2},\,\cdots,\,x_{N})\phi(x_{1},\,x_{2},\,\cdots,\,x_{N})=0
\end{equation}
which implies
\begin{equation}
\mathcal{O}(x_{1},\,x_{2},\,\cdots,\,x_{N})\psi(x_{\sigma_{1}},\,x_{\sigma_{2}},\,\cdots x_{\sigma_{N}})=0\label{eq:O-psi-sig}
\end{equation}
for $x_{\sigma_{1}}<x_{\sigma_{2}}<\cdots<x_{\sigma_{N}}$. Eq.~(\ref{eq:O-psi-sig})
implies that 
\begin{equation}
\mathcal{O}(x_{1},\,x_{2},\,\cdots,\,x_{N})\psi(x_{1},\,x_{2},\,\cdots,\,x_{N})=0\label{eq:O-psi}
\end{equation}
for $x_{1}<x_{2}<\cdots<x_{N}$. However, our goal is to show it holds
in all the regions. To reach this goal, let us take advantage of permutation
invariance. Note that r.h.s. of Eq.~(\ref{eq:O-psi-sig}) implies
the $x_{i}$'s on l.h.s. are dummy indices. Therefore, Eq.~(\ref{eq:O-psi-sig})
is equivalence as 
\begin{equation}
\mathcal{O}(x_{\sigma_{1}^{-1}},\,x_{\sigma_{2}^{-1}},\,\cdots,\,x_{\sigma_{N}^{-1}})\psi(x_{1},\,x_{2},\,\cdots x_{N})=\mathcal{O}(x_{1},\,x_{2},\,\cdots,\,x_{N})\psi(x_{1},\,x_{2},\,\cdots x_{N})=0
\end{equation}
where $\sigma^{-1}$ is the inverse permutation of $\sigma$. So far,
we have shown that Eq.~(\ref{eq:O-psi}) holds on \textit{all} the allowed regions  of $(x_{1},\,x_{2},\,\cdots x_{N})$.
On the boundary where at least two of the coordinates coincide, locality
of $\mathcal{O}$ implies that the value of $\mathcal{O}\psi$ on
the boundary is full determined by the its value on the neighborhood
outside the boundary, where $\mathcal{O}\psi$ vanishes. Thus, $\mathcal{O}\psi$
must also vanishes on the boundary.
\end{proof}

It can be shown that~\citep{Polychronakos92} 
\begin{equation}
[\pi_{i},\,\pi_{j}]=\sum_{k\neq i,\,j}V_{ijk}(M_{ijk}-M_{jik}).
\end{equation}
Thus $(M_{ijk}-M_{jik})\mathcal{P}_{\zeta}=0$ and we obtain
\begin{equation}
[\pi_{i},\,\pi_{j}]\mathcal{P}_{\zeta}=0.
\end{equation}
Furthermore $[\pi_{i},\,\pi_{j}]$ is local since they only involve
at most second derivative of the coordinate. However, it does not
vanish on the whole Hilbert space $\mathcal{H}^{\otimes N}$ because
$[\pi_{i},\,\pi_{j}]$ is not permutation-invariant. So here we can
see for a local, but not permutation-invariant operator, even it vanishes
on $\mathcal{H}_{\zeta}$, it may not vanish on $\mathcal{H}^{\otimes N}$.

\subsection{Effective $1$-body operator}

The operator
\begin{equation}
W_{i}=\text{i}\sum_{k\neq i}V_{ik}M_{ik}
\end{equation}
acts an effectively $1$-body potential, in the sense that it satisfies
the exchange rule for $1$-body operators, even if it involves $N$
degrees of freedom! Indeed,
\begin{align}
M_{ij}W_{j} & =\text{i}\sum_{k\neq j}M_{ij}V_{jk}M_{jk}=\text{i}\sum_{k\neq i,\,j}M_{ij}V_{jk}M_{jk}+\text{i}M_{ij}V_{ji}M_{ji}\nonumber \\
 & =\text{i}\sum_{k\neq i,\,j}V_{ik}M_{ik}M_{ij}+\text{i}V_{ij}M_{ij}M_{ij}\nonumber \\
 & =\text{i}\sum_{k\neq i}V_{ik}M_{ik}M_{ij}=W_{i}M_{ij}.\label{eq:effective1-pi1}
\end{align}
Further, 
\begin{align}
M_{ij}W_{k} & =\text{i}\sum_{l\neq k}M_{ij}V_{kl}M_{kl}=\text{i}\sum_{l\neq k,\,i,\,j}M_{ij}V_{kl}M_{kl}+\text{i}M_{ij}V_{ki}M_{ki}+\text{i}M_{ij}V_{ki}M_{ki}\nonumber \\
 & =\text{i}\sum_{l\neq k,\,i,\,j}V_{kl}M_{kl}M_{ij}+\text{i}V_{ki}M_{ki}M_{ij}+\text{i}V_{kj}M_{kj}M_{ij}\nonumber \\
 & =\text{i}\sum_{l\neq k}V_{kl}M_{kl}M_{ij}=W_{k}M_{ij}.\label{eq:effective1-pi2}
\end{align}
Therefore, $\pi_{i}$, $a_{i}$, $a_{i}^{\dagger}$ and $h_{i}$ act all
as effectively $1$-body operators, i.e., they satisfy the exchange
rules of a $1$-body operator.

{
\section{Guide to discover new integrable systems}

The systematic investigation of models in the PHJ family~\cite{CalogeroMarchioro75,Gambardella75,delcampo20,BeauDC21}
has been limited to date given that the resulting models were expected
to be, at best, quasi exactly solvable, i.e., models in which only
part of the spectrum can be determined. Our work establishes the equivalence
between the PHJ and EOF family and proves the integrability of the
PHJ models by identifying the corresponding integrals of motion. The
identification of new integrable models in this family is straightforward.
It suffices to choose a prepotential $V(x)$ or pair function $f(x)$,
compute their derivatives, and use Eq.~(\ref{eq:EOFH}) or Eq.~(\ref{eq:PHJH}) in the main text to determine the Hamiltonian of the model in the
real line. Similarly, Eq.~(\ref{eq:EOFH-embed}) or Eq.~(\ref{eq:PHJH-embed})
in the main text readily gives the Hamiltonian of the model in the
presence of an external potential. For the convenience of the reader,
we shall give the general integrable Hamiltonian in the EOF and PHJ
family explicitly and then provide a user guide to discover new integrable
systems. The general integrable Hamiltonian reads
\begin{equation}
H=\sum_{i}\left(\frac{p_{i}^{2}}{2m}+U_{i}\right)+\frac{1}{m}\left[\sum_{i<j}\left(\zeta\hbar V'_{ij}+V_{ij}^{2}-\zeta\sqrt{2m}V_{ij}[W_{i}-W_{j}]\right)-\sum_{i<j<k}V_{ijk}\right],\label{eq:IntegrableH}
\end{equation}
or in term of the Jastrow wave function it reads
\begin{equation}
H=\sum_{i}\left(\frac{p_{i}^{2}}{2m}+U_{i}\right)+\frac{\hbar^{2}}{m}\left[\sum_{i<j}\left(\frac{f_{ij}''}{f_{ij}}+(v_{i}'-v_{j}')\frac{f'_{ij}}{f_{ij}}\right)+\sum_{i<j<k}\left(\frac{f_{ij}'f_{ik}'}{f_{ij}f_{ik}}-\frac{f_{ij}'f_{jk}'}{f_{ij}f_{jk}}+\frac{f_{ik}'f_{jk}'}{f_{ik}f_{jk}}\right)\right],\label{eq:IntegrableH-f}
\end{equation}
where $V_{ij}=V(x_{i}-x_{j})$ is an odd function, $U_{i}=U(x_{i})$,
$W_{i}=W(x_{i})$, $V_{ijk}=V_{ij}V_{jk}+V_{jk}V_{ki}+V_{ki}V_{ij}$, and 
\begin{align}
V(x) & =\zeta\hbar\frac{f^{\prime}(x)}{f(x)},\label{eq:V-f-SM}\\
W(x) & =-\frac{\hbar}{\sqrt{2m}}v'(x),\\
U(x) & =W^{2}(x)-\frac{\hbar}{\sqrt{2m}}W'(x),\\
 & =\frac{\hbar^{2}}{2m}\left[(v_{i}')^{2}+v_{i}''\right].
\end{align}
The wave function associated with Eq.~(\ref{eq:IntegrableH}) is
\begin{align}
\Psi_{0} & =\prod_{i}\exp(v_{i})\prod_{i<j}f_{ij}\label{eq:WF-integrable}\\
 & =\exp\left(-\frac{\sqrt{2m}}{\hbar}\sum_{i}\int^{x_{i}}W(y)dy\right)\prod_{i<j}\exp[\hbar\int^{x_{ij}}dyV(y)].\label{eq:WF-VW}
\end{align}

In some situations, one may prefer to choose the pair function $f(x)$
or the prepotential $V(x)$ such that the three-body potential $V_{ijk}$
reduces to two-body potential. It has been shown by Calogero~\cite{Calogero75}
that in this case $V(x)$ must takes the form 
\begin{equation}
V(x)=\alpha\zeta(x;\,g_{2},g_{3})+\beta x,\label{eq:W-zeta}
\end{equation}
where $\alpha$, $\beta$, $g_{2}$, $g_{3}$ are constants and $\zeta(x;\,g_{2},g_{3})$
is the Weierstrass zeta-function~(see e.g., Chap 13 of \cite{erdelyi2007highertranscendental}).
In this case, for $x+y+z=0$, one can find
\begin{equation}
V(x)V(y)+V(y)V(z)+V(z)V(x)=Q(x)+Q(y)+Q(z),\label{eq:Three-Two}
\end{equation}
where 
\begin{equation}
Q(x)=\frac{1}{2}\left[\alpha\beta-\alpha V^{\prime}(x)-V^{2}(x)\right].\label{eq:M-zeta}
\end{equation}
One check that $V(x)$ in Eq.~(\ref{eq:W-zeta}) is odd in $x$ and thus $Q(x)$
is even in $x$. Therefore ~(\ref{eq:IntegrableH}) becomes~\cite{Sutherland71pra}
\begin{equation}
H=\sum_{i}\frac{p_{i}^{2}}{2m}+\sum_{i}U_{i}+\frac{1}{m}\sum_{i<j}\left[\zeta\hbar V'_{ij}+V_{ij}^{2}-(N-2)Q_{ij}-\zeta\sqrt{2m}V_{ij}(W_{i}-W_{j})\right],\label{eq:H-integrable-2body}
\end{equation}
where $N$ is the number of particles. The general solution Eq.~(\ref{eq:W-zeta})
immediately implies the following property of $V(x)$:
\begin{lem}
\label{lem:lin-gen}If $V(x)$ satisfies~(\ref{eq:Three-Two}), so
does $V(x)-\beta x$ \cite{Calogero75}. That is,
\begin{align}
V(x) & \to V(x)-\beta x,\label{eq:lin-gen-V}\\
Q(x) & \to Q(x)+\beta xV(x)-\frac{1}{2}\beta^{2}x^{2}.\label{eq:lin-gen-Q}
\end{align}
\end{lem}

Now we are in a position to give the guidelines for constructing an
integrable Hamiltonian:
\begin{enumerate}
\item As we have mentioned in the main text, the Hamiltonian~(\ref{eq:IntegrableH})
is always integrable regardless of the normalizability of $f(x)$.
When $f(x)$ is not normalizable, the Jastrow wave function~(\ref{eq:WF-integrable})
is no longer the ground state wave function of Eq. (\ref{eq:IntegrableH}).
However, Eq. (\ref{eq:IntegrableH}) is always legitimate.
\item The only requirement on $f(x)$ is that it has well-defined parity; it should be an even or odd function. 
Both parities lead to an odd prepotential $V(x)$. One can
either choose $f(x)$ or $V(x)$ as the starting point for constructing
the Hamiltonian.
\item When restricting particles in a ring with circumference $L$, all the potentials
must be periodic in $L$, which dictates that $f(x)$ or $V(x)$ must
also have period $L$. $W(x)$ and $v(x)$ are also required to have
the periodicity $L$. 
\end{enumerate}
Following the above guidelines, one should be able to construct an unlimited number of integrable models, as we now illustrate with examples. In
what follows, we shall use the notation $x_{ij}=x_{i}-x_{j}$ and
$\bar{x}_{ij}=(x_{i}+x_{j})/2$ so that the many-body Hamiltonian is written
in a compact way.

\subsection{Recovering the integrability of well-known models}

Now we show that the PHJ-EOF family includes canonical
examples of one-dimensional integrable systems. Thus their integrability
is proven immediately with Theorem~\ref{thm:quantum-integrability}
in the main text.

\subsubsection{Calogero model with and without a trap}

We shall employ Eq.~(\ref{eq:IntegrableH-f}) to generate the model
and consider

\begin{align}
v(x) & =-\frac{1}{2}\frac{m\omega}{\hbar}x^{2},\\
f(x) & =|x|^{\lambda},\lambda\in[0,\infty).\label{eq:f-calogero}
\end{align}
As we have mentioned in the guidelines, $f(x)$ is unnormalizable
since the configuration space is now unbounded. But the integrability
of the parent Hamiltonian is not affected. Here, we also would like
to emphasize that since the absolute value $|\cdot|$ is involved,
care needs to be taken when taking derivatives. It is straightforward
to calculate  
\begin{equation}
\frac{d}{dx}\big|F(x)\big|=\frac{d}{dx}\left\{ F(x)\text{sgn}[F(x)]\right\} =F^{\prime}(x)\text{sgn}[F(x)]+2F^{\prime}(x)F(x)\delta[F(x)].
\end{equation}
 We note the following identities
\begin{align}
F(x)\delta[F(x)] & =0,\,x\neq\xi_{k},\\
\int F(x)F^{\prime}(x)\delta[F(x)]dx & =\sum_{k}\int z\delta(z)dz=0,
\end{align}
where $\xi_{k}$ is the zeros of $F(x)$ and $z=F(x)$. Thus, one can
think of $F(x)\delta[F(x)]=0$ at all points, as long as it
is not multiplied by a factor that is singular at the zeros of $F(x)$.
Therefore, we obtain 
\begin{equation}
\frac{d}{dx}\big|F(x)\big|=F^{\prime}(x)\text{sgn}[F(x)],\label{eq:D-AbsF}
\end{equation}
and using Eq.~(\ref{eq:D-AbsF}), we find 
\begin{equation}
f^{\prime}(x)=\lambda|x|^{\lambda-1}\text{sgn}(x).\label{eq:f1st-Cal}
\end{equation}
As for the second derivative of the pair function, its explicit evaluation yields
\begin{equation}
f^{\prime\prime}(x)=\lambda|x|^{\lambda}\left[\frac{(\lambda-1)}{|x|^{2}}+\frac{2\delta(x)}{|x|}\right].\label{eq:f2nd-Cal}
\end{equation}
According to Eq.~(\ref{eq:V-f-SM}), we find
\begin{equation}
V(x)=\zeta\hbar\frac{\lambda}{|x|}\text{sgn}(x)=\zeta\hbar\frac{\lambda}{x}.
\end{equation}
Furthermore, we notice that when $\lambda\neq1$
\begin{equation}
\frac{(\lambda-1)}{|x|^{2}}+\frac{2\delta(x)}{|x|}=\frac{1}{|x|^{2}}[(\lambda-1)+2|x|\delta(x)]=\frac{\lambda-1}{|x|^{2}},
\end{equation}
where we have used Eq.~(\ref{eq:D-AbsF}) to conclude that $|x|\delta(x)$
is negligible, in comparison  with the  $\lambda-1$ term, whenever the latter is not
zero. To summarize, 
\begin{equation}
f^{\prime\prime}(x)=\begin{cases}
\frac{\lambda(\lambda-1)}{|x|^{2}} & \lambda\in[0,\,1)\cup(1,\infty),\\
\frac{2\delta(x)}{|x|} & \lambda=1.
\end{cases}
\end{equation}
One can further check that the three-body potential vanishes and that the
long-range potential induced by the external harmonic potential is
also a constant. Thus we find that the parent Hamiltonian~(\ref{eq:IntegrableH-f})
reads
\begin{equation}
H=\sum_{i}\left(\frac{p_{i}^{2}}{2m}+\frac{1}{2}m\omega_{0}^{2}x_{i}^{2}\right)-E_0+\begin{cases}
\frac{\hbar^{2}}{m}\frac{\lambda(\lambda-1)}{|x_{ij}|^{2}} & \lambda\in[0,\,1)\cup(1,\infty),\\
\frac{\hbar^{2}}{m}\frac{2\delta(x_{ij})}{|x_{ij}|} & \lambda=1,
\end{cases}\label{eq:PHJ-Cal}
\end{equation}
where
\begin{equation}
E_{0}=\frac{N\hbar\omega}{2}+\frac{N(N-1)\lambda\hbar\omega}{2}.
\end{equation}

The first line of Eq.~(\ref{eq:PHJ-Cal}) in the literature is usually
referred to as the (rational) Calogero model~\cite{Calogero71}. To the best of our knowledge, the careful treatment of the Calogero model  in the case $\lambda=1$ was first discussed by Mathieu Beau in 2017 \cite{Beau17note}. 

\subsubsection{Tonks-Girardeau gas, describing 1D hard-core bosons in a trap}
The Tonks-Girardeau regime describes one-dimensional hardcore bosons \cite{Girardeau60,GWT01}.
Consider the value $\lambda=1$ in Eq.~(\ref{eq:f-calogero}). The interaction
between the particles according to Eq.~(\ref{eq:PHJ-Cal}) is set by $\delta(x_{ij})/|x_{ij}|$. 
This interaction is equivalent to a repulsive delta interaction with infinite
strength, which describes  the hard-core potential  in the Tonks-Girardeau gas. In this case, the
ground state wave function becomes according to Eq.~(\ref{eq:WF-integrable})
\begin{equation}
\Psi_{0}=\exp\left(-\frac{m\omega_{0}}{2\hbar}\sum_{i}x_{i}^{2}\right)\prod_{i<j}|x_{ij}|,
\end{equation}
in agreement with \cite{GWT01}. The corresponding ground state energy is $E_{0}=\frac{N\hbar\omega}{2}+\frac{N(N-1)\hbar\omega}{2}$.

\subsubsection{The Sutherland model: inverse-square interacting particles on a ring}

We shall employ Eq.~(\ref{eq:IntegrableH-f}) to generate the model
and start with 
\begin{align}
f(x) & =\bigg|\sin\left(\frac{\pi x}{L}\right)\bigg|^{\lambda},\label{eq:f-Sutherland}\\
v(x) & =0.
\end{align}
According to Eq.~(\ref{eq:D-AbsF}), we find 
\begin{equation}
f^{\prime}(x)=\frac{\lambda\pi}{L}\bigg|\sin\left(\frac{\pi x}{L}\right)\bigg|^{\lambda-1}\cos\left(\frac{\pi x}{L}\right)\text{sgn}\left[\sin\left(\frac{\pi x}{L}\right)\right],
\end{equation}
and 
\begin{equation}
V(x)=\zeta\hbar\frac{f^{\prime}(x)}{f(x)}=\zeta\hbar\frac{\lambda\pi}{L}\cot\left(\frac{\pi x}{L}\right).\label{eq:V-Sutherland}
\end{equation}
Eq.~(\ref{eq:IntegrableH-f}) becomes the Sutherland model for $\lambda\in[0,\,1)$,
\begin{equation}
H=\sum_{i}\frac{p_{i}^{2}}{2m}+\frac{\hbar^{2}}{m}\left(\frac{\pi}{L}\right)^{2}\left[\sum_{i<j}\frac{\lambda(\lambda-1)}{\sin^{2}\left(\frac{\pi x_{ij}}{L}\right)}\right]-E_0,
\end{equation}
with the limiting case $\lambda\to1$ reduces to the Tonks-Girardeau
gas on a ring, which is similar with Eq.~(\ref{eq:PHJ-Cal}), where  
\begin{equation}
E_0=\left(\frac{\pi}{L}\right)^{2}\frac{\lambda^{2}\hbar^{2}N(N^{2}-1)}{6m}.
\end{equation}

\subsubsection{Lieb-Liniger gas }

The Lieb-Liniger model reads \cite{LL63,L63}
\begin{equation}
H=\sum_{i}\frac{p_{i}^{2}}{2m}+\frac{\hbar^{2}}{m}\left[\sum_{i<j}2g\delta(x_{ij})\right]-E_0,\label{eq:H-LL}
\end{equation}
where 
\begin{equation}
E_0=-\frac{\hbar^{2}g^{2}N(N^{2}-1)}{6m}.
\end{equation}
We shall employ Eq.~(\ref{eq:IntegrableH-f}) to generate Eq.~(\ref{eq:H-LL})
and consider \cite{delcampo20},
\begin{align}
f(x) & =\exp(g|x|),\\
v(x) & =0,
\end{align}
 which leads to 
\begin{equation}
V(x)=g\text{sgn}(x).\label{eq:V-LL}
\end{equation}
Since $\text{sgn}(x)\text{sgn}(y)+\text{sgn}(y)\text{sgn}(z)+\text{sgn}(z)\text{sgn}(x)=-1$
for $x+y+z=0$, the three-body potential is a constant. The two-body
potential term is the delta potential. 

\subsection{Predict the integrability of existing quasi-solve models or new models}

Embedding the Lieb-Liniger gas in a harmonic trap will lead to the
long-range Lieb-Liniger model~\cite{delcampo20,BeauPittman20}.
However, only its ground state property is discussed previously and
its integrability is not clear in the previous literature. Theorems
proved in the main text show that the long-range Lieb-Liniger model
is not only quasi-solvable but also integrable. Below we give more
applications of these theorems and show that a variety of related
models are actually integrable. 

\subsubsection{Quadratic Long-Range LL model}

According to Lemma~\ref{lem:lin-gen}, one can shift Eq.~(\ref{eq:V-LL})
by a linear function and take
\begin{equation}
V(x)=\zeta\hbar[g\text{sgn}(x)-\beta x].\label{eq:V-QLL}
\end{equation}
Setting $W(x)=0$, we find 
\begin{equation}
Q(x)=\hbar^{2}\left[-\frac{g^{2}N(N-1)}{6}+g\beta|x|-\frac{1}{2}\beta^{2}x^{2}\right],
\end{equation}
and Eq.~(\ref{eq:IntegrableH}) becomes
\begin{equation}
H=\sum_{i}\frac{p_{i}^{2}}{2m}+\frac{\hbar^{2}}{m}\sum_{i<j}\left[2g\delta(x_{ij})-Ng\beta|x_{ij}|+\frac{N\beta^{2}}{2}x_{ij}^{2}\right]-E_0,
\end{equation}
where
\begin{equation}
E_{0}=\frac{\hbar^{2}}{m}\left[\frac{N(N-1)\beta}{2}-\frac{g^{2}N(N^{2}-1)}{6}\right].
\end{equation}

This Hamiltonian was first constructed by Calogero~\cite{Calogero75}
with the Jastrow wave function. Its integrability was not clear previously.
Theorem~\ref{thm:quantum-integrability} in the main text shows that
this Hamiltonian is also integrable.

\subsubsection{Long-Range Hybrid LL model}
Considering the prepotential~(\ref{eq:V-QLL}) together with 
\begin{equation}
W(x)=\sqrt{\frac{m}{2}}\omega x,
\end{equation}
one finds Eq.~(\ref{eq:H-integrable-2body}) becomes
\begin{align}
H&=\sum_{i}\frac{p_{i}^{2}}{2m}+\frac{1}{2}m\omega^{2}\sum_{i}x_{i}^{2}-E_0\\&+\frac{\hbar^{2}}{m}\sum_{i<j}\left[2g\delta(x_{ij})-g\left(N\beta+\frac{m\omega}{\hbar}\right)|x_{ij}|+\beta\left(\frac{N\beta}{2}+\frac{m\omega}{\hbar}\right)x_{ij}^{2}\right],\label{eq:LR-Hybrid-LL}
\end{align}
where
\begin{equation}
E_{0}=\frac{N\hbar\omega}{2}-\frac{\hbar^{2}g^{2}N(N^{2}-1)}{6m}.
\end{equation}
Eq.~(\ref{eq:LR-Hybrid-LL}) is a further generalization of the Long-Range
LL model discussed in the main text and Refs.~\cite{delcampo20,BeauPittman20},
which we discuss for the first time here. This model is also integrable, according
to Theorem~\ref{thm:quantum-integrability} in the main text. Interestingly,
upon taking 
\begin{equation}
\beta=-\frac{m\omega}{N\hbar},
\end{equation}
the long-range Coulomb interaction vanishes and Eq.~(\ref{eq:LR-Hybrid-LL})
simplifies to 
\begin{equation}
H=\sum_{i}\frac{p_{i}^{2}}{2m}+\frac{1}{2}m\omega^{2}\sum_{i}x_{i}^{2}+\frac{\hbar^{2}}{m}\sum_{i<j}2g\delta(x_{ij})-\frac{m\omega^{2}}{2N}\sum_{i<j}x_{ij}^{2}-E_0.
\end{equation}
The above Hamiltonian describes the Lieb-Liniger gas in a harmonic
trap with quadratic interaction between the particles. Again, it is
integrable according to Theorem~\ref{thm:quantum-integrability}
in the main text. 

The same holds true if one takes 
\begin{equation}
W(x)=-\sqrt{\frac{m}{2}}\omega x,
\end{equation}
which leads to an unnormalizable wave function according to Eq.~(\ref{eq:WF-VW}).
However, as we have mentioned the integrability of the Hamiltonian
is not affected by the normalization of the wave function. Thus, the
sign of $\omega$ does not matter. 

\subsubsection{Generalized Hyperbolic model}

Since we know $V(x)=\zeta\hbar\lambda a\coth(ax)$ generates the hyperbolic
model~\cite{Polychronakos92}. According to Lemma~\ref{lem:lin-gen},
we take
\begin{equation}
V(x)=\zeta\hbar[\lambda a\coth(ax)-bx],
\end{equation}
and find 
\begin{equation}
Q(x)=\hbar^{2}\left[\lambda abx\coth(ax)-\frac{1}{2}b^{2}x^{2}+\text{constant}\right].
\end{equation}
Following the guidelines, we find in the absence of external potential
(i.e., $W_{i}=0$), Eq.~(\ref{eq:H-integrable-2body}) becomes 
\begin{align}
H & =\sum_{i}\frac{p_{i}^{2}}{2m}+\frac{\hbar^{2}}{m}\sum_{i<j}\left[\frac{a^{2}\lambda(\lambda-1)}{\sinh^{2}(ax_{ij})}-\lambda Nabx_{ij}\coth(ax_{ij})+\frac{1}{2}Nb^{2}x_{ij}^{2}\right]-E_0,
\end{align}
where 

\begin{equation}
E_{0}=\frac{\hbar^{2}}{m}\left[\frac{bN(N-1)}{2}-\frac{\lambda^{2}a^2 N(N^{2}-1)}{6}\right].
\end{equation}

This model was discussed by Calogero~\cite{Calogero75} and was
only known as quasi-solvable previously, i.e., only the ground state
wave function is known. Here, Theorem~\ref{thm:quantum-integrability}
in the main text indicates that this Hamiltonian is also integrable. 

\subsubsection{Sutherland model in a trigonometric trap}

To the best of our knowledge, the possibility of embedding a periodic potential for the Sutherland
model has not been yet discussed in the literature. We take this step
with the PHJ-EOF approach. We shall employ Eq.~(\ref{eq:IntegrableH-f})
and take $f(x)$ to be Eq.~(\ref{eq:f-Sutherland}) and
\begin{equation}
v(x)=-\frac{m\omega}{2\hbar}\left(\frac{L}{\pi}\right)^{2}\sin^{2}\left(\frac{\pi x}{L}\right),
\end{equation}
which yields 
\begin{equation}
U(x)=-\frac{\hbar\omega}{2}\cos\left(\frac{2\pi x}{L}\right)-\frac{m\omega^{2}L^{2}}{16\pi^{2}}\cos\left(\frac{4\pi x}{L}\right)+\frac{m\omega^{2}L^{2}}{16\pi^{2}}.
\end{equation}
We have computed before $f^{\prime}(x)/f(x)$ in Eq.~(\ref{eq:V-Sutherland}).
According to Eq.~(\ref{eq:IntegrableH-f}), we find

\begin{align}
H & =\sum_{i}\frac{p_{i}^{2}}{2m}-\sum_{i}\left[\frac{\hbar\omega}{2}\cos\left(\frac{2\pi x_{i}}{L}\right)+\frac{m\omega^{2}L^{2}}{16\pi^{2}}\cos\left(\frac{4\pi x_{i}}{L}\right)\right]\nonumber \\
 & +\frac{\hbar^{2}}{m}\sum_{i<j}\left[\left(\frac{\pi}{L}\right)^{2}\frac{\lambda(\lambda-1)}{\sin^{2}\left(\frac{\pi x_{ij}}{L}\right)}-\frac{\lambda m\omega}{\hbar}\cos\left(\frac{\pi x_{ij}}{L}\right)\cos\left(\frac{2\pi\bar{x}_{ij}}{L}\right)\right]-E_0,\label{eq:Sutherland-Embedding}
\end{align}
where 
\begin{equation}
E_{0}=\left(\frac{\pi}{L}\right)^{2}\frac{\lambda^{2}\hbar^{2}N(N^2-1)}{6m}-\frac{Nm\omega^{2}L^{2}}{16\pi^{2}}.
\end{equation}
In the thermodynamic limit $L,\,N\to\infty$ with $N/L$ kept fixed,
so that one can readily check Eq.~(\ref{eq:Sutherland-Embedding})
reduces to the Calogero model in a harmonic trap. Thus Eq.~(\ref{eq:Sutherland-Embedding})
can be viewed as the generalization of the embedded Calogero model
to the embedded the Sutherland model. Again, if one flips the sign
of the frequency $\omega$, the integrability is preserved.

\subsection{Fixing interactions first: Toda-like interactions in the continuum}

The above examples have been found in the PHJ-EOF family by choosing the pair function $f(x)$.
As an alternative, one can twist the construction around, by fixing the interactions first.
For the sake of illustration, let us consider a two-body potential that decays exponentially with the interparticle distance over a length scale $\ell$, e.g., 
\beqa
V_2=\frac{\hbar^2g}{m}\sum_{i< j}e^{-|x_{ij}|/\ell}.
\eeqa
This potential is a generalization to the continuum of the Toda interactions and can also be considered as a low-density approximation to the hyperbolic potential, e.g.,  as discussed in \cite{Sutherland04,delcampo20}.
Thus, we expect 
\begin{equation}
\frac{f''(x)}{f(x)}=ge^{-|x|/\ell}+\text{singular terms}. 
\end{equation} 
Ignoring the singular terms for the moment, this differential equation for $x>0$ admits as a specific solution 
\begin{equation}
f(x)=I_{0}\left(2\sqrt{g}le^{-x/\ell}\right),
\end{equation} 
 where $I_{\alpha}(x)$ is the modified Bessel functions of first kind and order $\alpha$. This solution motivates the choice
 \begin{equation}
f(x)=I_{0}\left(2\sqrt{g}le^{-|x|/\ell}\right)
\end{equation}
for all $x$.
 This pair function decays smoothly as function of $x$ to unit value. To compute $f''(x)/f(x)$, care must be taken for the absolute value. Using Eq.~(\ref{eq:D-AbsF}), we can readily calculate  
\begin{align}
\frac{dI_{0}\left(2\sqrt{g}le^{-|x|/\ell}\right)}{dx} & =\frac{dI_{0}\left(2\sqrt{g}le^{-|x|/\ell}\right)}{d|x|}\text{sgn}(x)\nonumber \\
 & =-2e^{-|x|/\ell}\sqrt{g}I_{1}\left(2\sqrt{g}le^{-|x|/l}\right)\text{sgn}(x),\\
\frac{d^{2}I_{0}\left(2\sqrt{g}le^{-|x|/\ell}\right)}{dx^{2}} & =\frac{d^{2}I_{0}\left(2\sqrt{g}le^{-|x|/\ell}\right)}{d^{2}|x|}+\frac{dI_{0}\left(2\sqrt{g}le^{-|x|/\ell}\right)}{d|x|}\delta(x)\nonumber \\
 & =\frac{d^{2}I_{0}\left(2\sqrt{g}le^{-|x|/\ell}\right)}{d^{2}|x|}-2e^{-|x|/l}\sqrt{g}I_{1}\left(2\sqrt{g}le^{-|x|/\ell}\right)\delta(x).
\end{align}
Thus we find 
\begin{align}
\frac{f'(x)}{f(x)} & =-2\sqrt{g}e^{-|x|/\ell}\frac{I_{1}\left(2\sqrt{g}le^{-|x|/\ell}\right)}{I_{0}\left(2\sqrt{g}le^{-|x|/\ell}\right)}\text{sgn}(x),\label{eq:fprime-f}\\
\frac{f''(x)}{f(x)} & =ge^{-|x|/\ell}-c\delta(x)\label{eq:f2prime-f},\\
\end{align}
where 
\begin{equation}
c=\frac{2\sqrt{g}I_{1}\left(2\sqrt{g}\ell\right)}{I_{0}\left(2\sqrt{g}\ell\right)}.
\end{equation}
 
The corresponding three-body potential is nonzero, as expected, and takes the form given in Eq.~(\ref{eq:IntegrableH-f}) with $f'(x)/f(x)$ given by Eq.~\eqref{eq:fprime-f}.
The expression of $f'(x)/f(x)$ sets the generalized momenta $\pi_i$ through Eq.~\eqref{eq:Gpi-f} in the main text and the corresponding integrals of motion $I_n$. 
The resulting Hamiltonian thus takes the form
\beqa
H=\sum_{i}\frac{p_{i}^{2}}{2m}+\frac{\hbar^{2}}{m}\sum_{i< j}\left[ge^{-|x_{ij}|/\ell}-c\delta(x_{ij})\right]+V_3.
\eeqa

Given $I(0)=1$, i.e., the wave function does not decay when particles are far apart, the Jastrow wave function is not normalizable in the in absence of an external trap. We thus consider the case where system is trapped. The ground state of the trapped system is then 
\beqa 
\Psi_0=e^{-\frac{m\om}{2\hbar}\sum_{i=1}^Nx_i^2}\prod_{i<j}I_0[2\ell\sqrt{g}\exp(-|x_{ij}|/\ell)].
\eeqa
The long-range two-body potential due to the embedding of the external harmonic trap in Eq.~(\ref{eq:IntegrableH-f}) is similarly given in terms of the modified Bessel functions.\\

}
 
\end{document}